\documentclass[10pt,conference]{IEEEtran}
\usepackage{xcolor,soul,framed} 
\usepackage[noadjust]{cite}
\usepackage[
bookmarks=false,
hidelinks
]{hyperref}
 
\colorlet{shadecolor}{yellow}

\usepackage[pdftex]{graphicx}
\graphicspath{{../pdf/}{../jpeg/}}
\DeclareGraphicsExtensions{.pdf,.jpeg,.png}
\usepackage[caption=false, font=footnotesize]{subfig}
\usepackage{amssymb,mathtools}
\usepackage{array}
\usepackage{mdwmath}
\usepackage{mdwtab}
\usepackage{eqparbox}
\usepackage{url}

\usepackage{enumitem}

\usepackage[ruled,vlined,linesnumbered]{algorithm2e}
\usepackage{enumitem}
\usepackage[mathcal]{euscript} 
\usepackage{tabularx} 
\usepackage{multirow}
\usepackage{booktabs}
\usepackage{makecell} 
\usepackage{aligned-overset}
\usepackage{multicol}
\usepackage{balance}
\usepackage{bbm}
\usepackage{amsmath}
\usepackage{amsthm}
\newtheoremstyle{mytheorem}
  {0pt}
  {0pt}
  {\normalfont} 
  {}
  {\bfseries\itshape} 
  {.}
  {0.5em} 
  {} 
\theoremstyle{mytheorem}
\usepackage[top=0.7in, bottom=1in, left=0.625in, right=0.625in]{geometry}
\usepackage{pifont}

\graphicspath{{./Figures/}}
\newtheorem{Remark}{\textit{Remark}}

\newtheorem{Theorem}{Theorem}

\renewcommand{\vec}[1]{\boldsymbol{\mathrm{#1}}}

\newcommand{\IRS}{\mathtt{RIS}}
\newcommand{\RIS}{\mathtt{RIS}}

\newcommand{\user}{\mathtt{UE}}

\newcommand{\ap}{\mathtt{AP}}

\newcommand{\pl}{\ell}
\usepackage[hidelinks]{hyperref}
\usepackage{xcolor}

\begin{document}
      \title{RIS-Assisted Reliability Maximization for URLLC in Upper Mid-Band 6G Networks}
\author{
  \IEEEauthorblockN{Muddasir Rahim, Irfan Azam, and Soumaya Cherkaoui}
  \IEEEauthorblockA{Department of Computer Engineering and Software Engineering, Polytechnique Montreal, Canada} \IEEEauthorblockA{Emails: muddasir.rahim@polymtl.ca, irfan.azam@polymtl.ca, soumaya.cherkaoui@polymtl.ca}
}

\maketitle
\begin{abstract}
In sixth-generation (6G) networks, supporting ultra-reliable low-latency communication (URLLC) services imposes stringent requirements on reliability, latency, and efficient resource utilization. The upper mid-band (FR3) spectrum has emerged as a promising candidate due to its favorable trade-off between coverage and bandwidth availability. However, achieving high reliability in dense deployments remains challenging due to interference and propagation impairments. To address these challenges, this paper investigates a reconfigurable intelligent surface (RIS)-assisted network operating in the FR3 band. We formulate a joint transmit power allocation and user equipment (UE)-RIS association problem to maximize network reliability under finite blocklength constraints. The resulting problem is nonconvex and combinatorial due to interference coupling and binary association variables. To solve this problem efficiently, we propose a multiphase resource allocation framework that integrates a successive convex approximation (SCA)-based power optimization scheme with a matching theory-based UE-RIS association algorithm. Simulation results show that the proposed scheme outperforms greedy and random baselines and achieves near-optimal performance close to the exhaustive search scheme.
\end{abstract}
\section{INTRODUCTION}\label{introduction}
\IEEEPARstart{T}{he} rapid evolution of wireless applications has introduced diverse and demanding requirements for futuristic wireless communication networks, including ultra-high data rates, low latency, and high reliability~\cite{ranjha2022urllc,rahim2026reliable}. In particular, ultra-reliable low-latency communication (URLLC) services require strict reliability guarantees under limited latency budgets, which poses significant challenges in dense network deployments with constrained power resources and strong interference. Consequently, sixth-generation (6G) networks must support massive connectivity while ensuring dependable communication for heterogeneous devices and emerging data-intensive applications such as extended reality (XR) and ultra-high-definition video streaming \cite{pala2023spectral}.

To meet these requirements, the exploration of new spectrum bands has become essential. Among the candidate frequency ranges, the upper mid-band (UMB), also referred to as frequency range 3 (FR3) (7-15 GHz), has recently gained significant attention due to its favorable trade-off between bandwidth availability and coverage~\cite{chaves2024coverage}. Compared with conventional sub-6 GHz bands, FR3 offers wider bandwidths that enable higher data rates, while avoiding the severe propagation losses associated with millimeter-wave and terahertz frequencies~\cite{rahim2026dual}. As a result, FR3 is increasingly recognized as a key enabler for 6G systems. However, achieving reliable communication in FR3 remains challenging due to non-negligible path loss and frequent line-of-sight (LoS) blockages in cluttered environments, particularly in dense network deployments. These limitations motivate the adoption of environment-aware technologies capable of dynamically shaping the wireless propagation environment. Reconfigurable intelligent surfaces (RISs) have emerged as a promising technology to address these challenges by enabling programmable wireless environments. Specifically, RISs consist of a large number of passive reflecting elements that can intelligently manipulate the propagation of electromagnetic waves, thereby enhancing signal quality and extending coverage~\cite{rahim2024jpusa}. While RIS-assisted communications have been extensively studied in millimeter-wave and terahertz bands, their integration into FR3 systems remains relatively underexplored. This gap is particularly pronounced in URLLC scenarios, where a large number of low-power devices compete for limited resources, and reliability becomes a critical performance metric.
 
In this context, increasing research efforts have focused on investigating the performance limits and design strategies for FR3-based networks. However, studies specifically addressing MIMO techniques in the FR3 band remain relatively limited~\cite{heath2024beamsharing,bjornson2025enabling,tian2024mid}. For instance, the work in~\cite{heath2024beamsharing} considered a MIMO system supporting both near-field and far-field users and proposed a user pairing strategy combined with beam sharing to accommodate heterogeneous users. In~\cite{bjornson2025enabling}, the term gigantic MIMO (gMIMO) is used to refer to arrays with at least 256 antenna ports, emphasizing deployments that go beyond typical ultra-massive MIMO configurations. Furthermore,~\cite{tian2024mid} investigated extra large-scale MIMO (XL-MIMO) for mid-band communications and developed analytical models to evaluate key performance metrics such as spectral efficiency and outage probability.

In parallel, RISs have been widely studied as an enabling technology for enhancing wireless communication performance. The survey in~\cite{mohsan2023irs} provides a detailed overview of RIS-assisted UAV communications, including system architectures, channel modeling, and common design goals. The authors in~\cite{abdalla2020uavs} further explored the integration of RISs with UAV platforms, highlighting potential applications and open research issues in future wireless networks. A more design-oriented contribution is presented in~\cite{bui2025joint}, where a joint optimization framework was proposed for transmit power allocation and RIS phase-shift tuning, aiming to improve energy efficiency while ensuring reliable service for all users. Meanwhile,~\cite{farre2025dynamic} examined the integration of RISs in combined non-terrestrial and terrestrial networking, targeting improved connectivity and capacity in dense environments by optimizing coverage as well as interference suppression. The work in~\cite{you2022enabling} also introduced two complementary schemes: RIS-assisted UAV links and aerial-RIS-assisted terrestrial links, illustrating how UAVs and RISs can be jointly exploited to support integrated aerial-terrestrial communications in next-generation networks.

Despite these advances, RIS research in FR3 is still in its early stages. To the best of the authors’ knowledge,~\cite{kara2024reconfigurable} is one of the few works that explicitly demonstrates the relevance of RISs in the FR3 band. Overall, the RIS literature remains heavily concentrated on FR2 (mmWave/THz) and primarily single-tier terrestrial scenarios, while multi-tier operation in FR3, including the combined use of aerial and terrestrial RISs, has received comparatively little attention, even though FR3 is becoming increasingly important for 6G and exhibits different coverage and blockage trade-offs.

\textbf{Contributions:} Existing works in RIS-assisted networks typically optimize power allocation or user association independently, despite their interdependence. Moreover, Shannon capacity formulations are inadequate for URLLC with short packets and strict reliability. To address these limitations, we propose a unified framework for reliability-oriented resource allocation in RIS-assisted FR3 networks. The main contributions are summarized as follows:
\begin{itemize}
  \item We develop a RIS-assisted system model operating in the FR3 band, explicitly incorporating finite blocklength communication to capture URLLC reliability requirements. Based on this model, we formulate a joint transmit power allocation and UE-RIS association problem to maximize network reliability. 
  \item To solve the resulting nonconvex and combinatorial problem, we propose a multiphase framework that integrates a successive convex approximation (SCA)-based power optimization scheme with a matching-theory-based UE-RIS association algorithm, enabling efficient and scalable implementation.
  \item Simulation results demonstrate that the proposed scheme outperforms greedy and random baselines and achieves near-optimal performance close to the exhaustive search scheme.
\end{itemize}
\section{System Model} \label{model}
We consider a network consisting of a single AP equipped with $N$ antennas serving the set of user equipment (UE) $\mathcal{K} = \{\user_1, \user_2, \ldots, \user_k, \ldots, \user_K\}$ supporting URLLC services. The UEs receive their signals either through direct links from the AP or via one of the available RISs when the direct path is weak or blocked. The channels between the AP and the RISs, as well as between the RISs and the UEs, are modeled as quasi-static fading channels in the FR3 band. To enhance reliability in environments with potential blockages, the network is assisted by $L$ RISs. The RIS set is denoted by $\mathcal{L} = \{\RIS_1, \RIS_2, \ldots, \RIS_l, \ldots, \RIS_L\}$. Each RIS comprises $M = M_y M_z$ passive reflectors arranged in a uniform planar array, where $M_y$ and $M_z$ denote the number of reflectors along the horizontal and vertical directions, respectively. The distance from the AP to  $k^{th}$ UE ($\user_{k}$) is $d_{a, k}$, from the AP to the center of $l^{th}$ RIS ($\RIS_{l}$) is $d_{a, l}$, and from $\RIS_{l}$ to $\user_{k}$ is denoted as $d_{l, k}$.
\begin{figure}[!t]
     \centering
\includegraphics[width=0.9\linewidth]{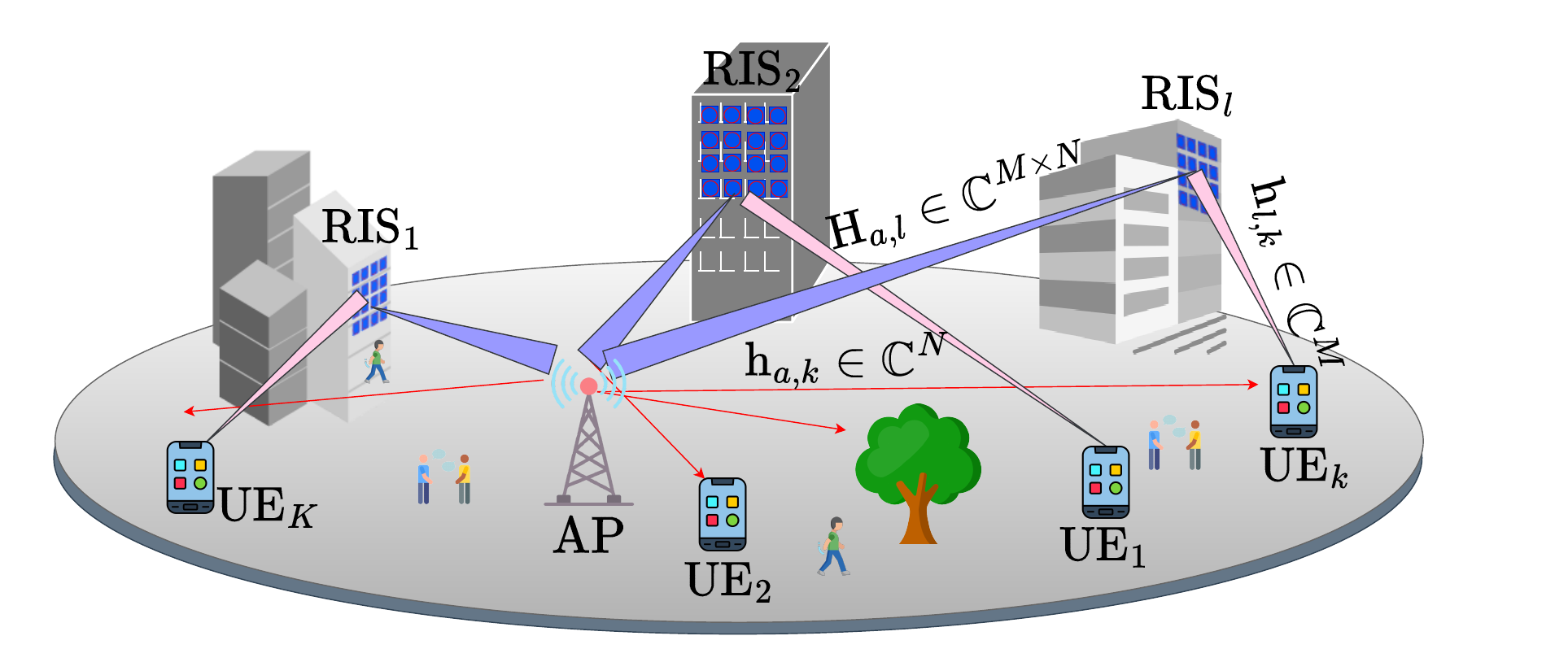}
    \caption{System model of RIS-assisted URLLC network in FR3 band.
     }
     \label{m1}   
\end{figure}
\subsubsection{\textbf{Channel Model}}
We consider a downlink communication scenario that includes both the direct link and RIS-assisted reflected links. Let $\vec{h}_{a,k}\in \mathbb{C}^{N}$  denote the direct channel vector from $\ap$-to-$\user_k$, where the channel coefficient between $n^{th}$ antenna of the AP ($\ap_n$)-to-$\user_k$ is $\textstyle h_{a,n,k} = \sqrt{\pl_{a,n,k}} \, e^{-j \omega d_{a,k}},$
where $\pl_{a,n,k}$ represents the pathloss between $\ap_n$ and $\user_{k}$ and $\omega = \frac{2\pi f_c}{c}$ is the wave number at the carrier frequency $f_c$. The channel matrix between $\ap$ and $\RIS_l$ is denoted by $\vec{H}_{a,l}\in \mathbb{C}^{M \times N}$ is, where the channel coefficient between $\ap_n$ and $m^{th}$ reflecting element of $\RIS_l$ ($\RIS_{l,m}$) is $\textstyle  [\vec{H}_{a,l}]_{n,m} =h_{n,l,m} = \sqrt{\pl_{n,l,m}} \, e^{-j \omega d_{a,l}},$
where $\pl_{n,l,m}$ represents the pathloss from $\ap_n$-to-$\RIS_{l,m}$. Finally, the channel vector from $\RIS_{l}$-to-$\user_k$ is denoted by $\vec{h}_{l,k} \in \mathbb{C}^{M}$, where the channel coefficient corresponding to $\RIS_{l,m}$ is given by $\textstyle  [\vec{h}_{l,k}]_m= h_{l,m,k} = \sqrt{\pl_{l,m,k}} \, e^{-j \omega d_{l,k}},$
where $\pl_{l,m,k}$ is the pathloss from $\RIS_{l,m}$-to-$\user_k$. Let $\vec{\Omega}$ denote the UE-RIS association matrix, where $\Omega_{k,l}\in\{0,1\}$ is a binary association variable indicating whether $\user_k$ is associated with $\IRS_l$.
Since each UE is associated with at most one RIS, we define $\vec{h}_{a,l,k}\in \mathbb{C}^{N}$ as the effective channel vector from $\ap$-to-$\user_k$ corresponding to the direct path and the reflected path via $\IRS_l$, which can be written as
\begin{align}\label{eq_channel}
  \textstyle  \vec{h}_{a,l,k} =\textstyle\vec{h}_{a,k} +\textstyle \Omega_{k,l} \vec{H}_{a,l}^{\sf H} \, \vec{\Theta}_l \, \vec{h}_{l,k},
 \end{align}
 where $\vec{\Theta}_l = \mathrm{diag}([\kappa_{l,1} e^{j \theta_{l,1}}, \ldots,\kappa_{l,m} e^{j \theta_{l,m}}, \ldots, \kappa_{l,M} e^{j \theta_{l,M}}])$ denotes the RIS reflection matrix of $\IRS_l$ and $\kappa_{l,m} \in [0,1]$ and $\theta_{l,m} \in [0, 2\pi)$ represent amplitude reflection coefficients and phase shift of the $m^{th}$ element, respectively. The RIS phase shifts are assumed fixed based on a statistical channel state information (CSI) design, where configurations are optimized over large-scale channel parameters and updated at a slower timescale. This reduces signaling overhead and is suitable for URLLC with strict latency constraints. 
\subsubsection{\textbf{Signal Model and Data Rate of URLLC}}
We consider linear precoding at the AP, where the transmit beamforming matrix is denoted by ${\vec{W}} \in \mathbb{C}^{N \times K}$, with ${\vec{w}}_k \in \mathbb{C}^{N}$ representing the beamforming vector associated with $\user_k$. Each beamforming vector is expressed as ${\vec{w}}_k = \sqrt{p_k}\widehat{\vec{w}}_k,$ where $p_k$ denotes the transmit power allocated to $\user_k$ and $\widehat{\vec{w}}_k$ is the corresponding unit-norm beamforming direction, i.e., $\Vert\widehat{\vec{w}}_k\Vert^2=1$. Given a total  power budget $P_{\max}$, the power constraint is $\sum_{k=1}^K \Vert{\vec{w}}_k\Vert^2 =\sum_{k=1}^K p_k\leq P_{\max}.$ 
Let $s_k$ denote the information symbol intended for $\user_k$ and $\vec{x} \in \mathbb{C}^{N}$ be the transmit signal vector from the AP to UEs, which is given by 
\begin{align} \label{transmit_signal}
\textstyle\vec{x}=\sum_{{k=1 }}^K{\vec{w}}_k s_k=\sum_{{k=1  }}^K\sqrt{p_k}\widehat{\vec{w}}_k s_k.
\end{align}
The received signals at $\user_k$ is expressed as 
\begin{align}\label{reciver_1}
  \textstyle  y_{k,l}& = \sqrt{p_{k}}(\mathbf{h}_{a,l,k})^H \widehat{\vec{w}}_k s_k +  \textstyle\sum_{i \neq k}^{K}\sqrt{p_{i}}(\mathbf{h}_{a,l,k})^H \widehat{\vec{w}}_i s_i + n_k\nonumber\\ &=\underbrace{ \textstyle\sqrt{p_{k}}(\vec{h}_{a,k} +\Omega_{k,l} \vec{H}_{a,l}^{\sf H}  \vec{\Theta}_l \vec{h}_{l,k})^H \widehat{\vec{w}}_k s_k}_{\text{desired signal}}\nonumber\\ & + \underbrace{ \textstyle\sum_{i \neq k}^{K}\sqrt{p_{i}}(\vec{h}_{a,k} +\Omega_{k,l} \vec{H}_{a,l}^{\sf H}  \vec{\Theta}_l \vec{h}_{l,k})^H \widehat{\vec{w}}_i s_i}_{\text{interference}}\!+\! \underbrace{n_k}_{\text{noise}},
\end{align}
where $n_k\sim \mathcal{CN}(0, \sigma_k^2)$ represents the additive white Gaussian noise (AWGN). Accordingly, the signal-to-noise-plus-interference ratio (SINR) at $\user_k$ is given by
   \begin{align}\label{sinr}  \Lambda_{k,l}=\frac{p_k\Vert(\vec{h}_{a,k} +\Omega_{k,l} \vec{H}_{a,l}^{\sf H}  \vec{\Theta}_l \vec{h}_{l,k})^H \widehat{\vec{w}}_k \Vert^2 }{\!\!\sum_{i\neq k}^K\! p_i\big(\Vert(\vec{h}_{a,k}\!\!+\! \Omega_{k,l}\vec{H}_{a,l}^{\sf H}  \vec{\Theta}_l \vec{h}_{l,k})^H \widehat{\vec{w}}_i\Vert^2\big)\!+\!\sigma_k^2}.
 \end{align}
In URLLC, the well-known Shannon's data rate formula for long block-length wireless communications is inaccurate. Thus, we use the finite block length formula to calculate the achievable rate of URLLC UEs, which can be expressed as~\cite{ranjha2022urllc}
\begin{align}\label{rate_u}
  \textstyle  R_{k,l} & = \textstyle \log_2 \big( 1 + \Lambda_{k,l} \big)-\sqrt{\frac{V_{k,l}}{\zeta_{k}}}Q^{-1}(\epsilon_{k,l})+\frac{O(\log \zeta_{k})}{\zeta_{k}}\nonumber\\&  \textstyle\approx \log_2 \big( 1 + \Lambda_{k,l} \big)-\sqrt{\frac{V_{k,l}}{\zeta_{k}}}Q^{-1}(\epsilon_{k,l}),
\end{align}
where $\zeta_{k}$, $\epsilon_{k,l}$, and $V_{k,l}$ denote the packet length, error probability, and channel dispersion, respectively. The channel dispersion is defined as $V_{k,l}=1-{1}/{(1+\Lambda_{k,l})^{2}}.$ Furthermore, $Q^{-1}(\epsilon_{k,l})$ is the inverse Q-function of $\epsilon_{j}$. By rearranging \eqref{rate_u}, the decoding error probability can be expressed as follows:
\begin{align}\label{error_pr}
 \textstyle \epsilon_{k,l}=Q\Big( \sqrt{\frac{\zeta_{k}}{V_{k,l}}}\big(\log_2 ( 1+ \Lambda_{k,l} )-R_{k,l}\big)\Big).
\end{align}
 Accordingly, the reliability of $\user_k$ via $\IRS_l$ is defined as
\begin{align}\label{reliability}
 \Psi_{k,l} = 1- \epsilon_{k,l}.
\end{align}
\subsubsection{\textbf{Optimization Problem Formulation}}\label{pform}
In URLLC networks, stringent reliability and low-latency requirements introduce significant challenges in resource management, especially in interference-limited environments with dynamic network topology. To address these challenges, this work focuses on maximizing the network sum reliability by jointly optimizing the transmit power allocation and UE-RIS association. Accordingly, the optimization problem is formulated to enhance the reliability of all users under finite blocklength transmission as follows: 
\begin{subequations}\label{eq_opt_prob}
\begin{alignat}{2}
&\mathcal{P}: \underset{\vec{P, \,\Omega}}{\text{ max}}
& 
&\quad\textstyle\sum_{k=1}^K \Omega_{k,l}\Psi_{k,l} (\vec{P, \Omega}), \label{eq_optProb}\\ 
&\text{s.t.} 
 &&\mathrm{C}1:\quad  p_k \geq 0, \quad\forall k\in K, \label{eq_constraint1}\\
 &&&\mathrm{C}2:\quad 
 \textstyle \sum_{k=1}^K p_k\leq P_{\max},\label{eq_constraint2}\\
 &&&\mathrm{C}3:\quad  \Omega_{k,l}\in\{0,1\}, \ \forall k,l, \label{eq_constraint3}\\
 &&&\mathrm{C}4:\quad  \textstyle \sum_{l=1}^{L} \Omega_{k,l} \le 1,\ \forall k, \label{eq_constraint4}\\ 
 &&&\mathrm{C}5:\quad  \textstyle\sum_{k=1}^{K} \Omega_{k,l} \le 1,\ \forall l, \label{eq_constraint5}
\end{alignat}
\end{subequations}
where $P_{\max}$ is the maximum transmit power budget, constraints~$\mathrm{C}1$ and~$\mathrm{C}2$ ensure that the transmit power of all users is always positive and the total power does not exceed the power budget $P_{\max}$, respectively. Furthermore, constraint~$\mathrm{C}3$  enforces binary assignment to complete the one-to-one matching, constraint~$\mathrm{C}4$ ensures that each UE is associated with at most one RIS, while constraint~$\mathrm{C}5$ restricts each RIS to serve at most one UE.
\begin{Remark}\label{R1}
Problem $\mathcal{P}$ in \eqref{eq_opt_prob} is a mixed-integer nonconvex optimization problem due to the coupling between $\vec{P}$ and $\vec{\Omega}$. Even with fixed power allocation, the UE-RIS association subproblem remains combinatorial under \eqref{eq_constraint3}--\eqref{eq_constraint5}. Hence, exhaustive search becomes impractical for moderate network sizes, which motivates the proposed two-phase solution.
\end{Remark}
\section{Proposed Resource Allocation Framework}\label{proposed} 
Problem~\eqref{eq_opt_prob} is NP-hard as established in Remark~\ref {R1}. To address this, we adopt a sequential optimization framework in which the AP first determines the downlink transmit power allocation, followed by the optimization of UE-RIS association. This decoupled approach enables efficient resource allocation while enhancing the received signal quality and overall network reliability in the FR3 band.
\subsection{SCA-based Power Optimization}
Given the RIS association matrix $\vec{\Omega}$, we optimize the downlink transmit power allocation of UEs. The power optimization subproblem is given by
\begin{subequations}\label{eq_opt_prob1.2}
\begin{alignat}{2}
& \mathcal{P}1: \underset{\vec{\mathbf{P}}}{\text{ max}}
& 
&\textstyle\quad\quad\sum_{k=1}^K \Psi_{k,l} (\vec{ P\,|\, \Omega}), \label{eq_optProb1.2}\\ 
&\text{s.t.} 
 &&\quad\quad  \mathrm{C}1, \, \mathrm{C}2.\label{eq_constraint2.2}
\end{alignat}
\end{subequations}
The objective function in~\eqref{eq_opt_prob1.2} is nonconvex due to the interference-coupled SINR, even for the given $\vec{\Omega}$. Therefore, to address this issue, we use SCA to solve the subproblem. The SINR in~\eqref{sinr} can be written as
  \begin{align}\label{SCA_1}
     \Lambda_{k,l}& =\frac{p_k {g_{k,k}}}{\textstyle\sum_{ i\neq k}^K p_ig_{k,i}+\sigma_k^2}=\frac{p_k {g_{k,k}}}{I_k(\vec{P})},
 \end{align}  
where $I_k(\vec{P})=\sum_{ i\neq k}^K p_ig_{k,i}+\sigma_k^2$ is interference, $g_{k,k}=\Vert(\vec{h}_{a,k} +\Omega_{k,l} \vec{H}_{a,l}^{\sf H}  \vec{\Theta}_l \vec{h}_{l,k})^H \widehat{\vec{w}}_k  \Vert^2$ is the effective desired channel gain and $g_{k,i}=\Vert(\vec{h}_{a,k} +\Omega_{k,l} \vec{H}_{a,l}^{\sf H}  \vec{\Theta}_l \vec{h}_{l,k})^H \widehat{\vec{w}}_i\Vert^2$ is the interference channel gain  at $\user_k$. Using~\eqref{reliability}, maximizing the sum reliability is equivalent to minimizing the error probability. Thus, the objective function in~\eqref{eq_opt_prob1.2} is written as
\begin{align}
\textstyle \underset{\vec{\mathbf{P}}}{\text{ max}}\sum_{k=1}^K \Psi_{k,l} (\vec{P| \Omega})\approx \underset{\vec{\mathbf{P}}}{\text{ min}}\sum_{k=1}^K \epsilon_{k,l} (\vec{ P| \Omega}). \label{eq_optProb1.2.1}
\end{align}
Here, $\epsilon_{k,l}$ is a Q-function which is difficult to solve directly due to its nonlinear nature. Therefore, we introduce an auxiliary variable as
\begin{align}\label{variable}
    \textstyle \eta_{k,l}= \sqrt{{\zeta_{k}}/{V_{k,l}}}\big(\log_2 ( 1+ \Lambda_{k,l} )-R_{k,l}\big)
\end{align}
\begin{align}\label{variable_1}
  \textstyle   \epsilon_{k,l}=Q( \eta_{k,l}).
\end{align}
Since the Q-function in~\eqref{variable_1} is monotonically decreasing, minimizing $\epsilon_{k,l}$ is equivalent to maximizing $\eta_{k,l}$. Therefore, maximization in \eqref{eq_optProb1.2.1} can be approximated as
\begin{align}
\textstyle\underset{\vec{\mathbf{P}}}{\text{ max}}\sum_{k=1}^K \Psi_{k,l} (\vec{ P| \Omega})\approx\underset{\vec{\mathbf{P}}}{\text{ max}}\sum_{k=1}^K\eta_{k,l} (\vec{ P|\Omega}). \label{eq_optProb1.2.2}
\end{align}
To make the problem tractable, the channel dispersion is evaluated at the current SCA iterate and kept fixed for the next iteration. Accordingly, we rewrite $\eta_{k,l}$ as
\begin{align}\label{variable_2}
 \textstyle   \eta_{k,l}= \sqrt{{\zeta_{k}}/{V_{k,l}^{(t)}}}\big(\log_2 ( 1+ \Lambda_{k,l} )-R_{k,l}\big).
\end{align}
Since $R_{k,l}$ is constant, $\eta_{k,l}$ can be rewritten as $\textstyle  \eta_{k,l}\approx \sqrt{{\zeta_{k}}/{V_{k,l}^{(t)}}}\big(\log_2 ( 1+ \Lambda_{k,l} )\big).$
Now using~\eqref{SCA_1}, $\eta_{k,l}(\vec{P})$ can be equivalently rewritten as
{\small
\begin{align}\label{SCA_2}
\textstyle\eta_{k,l}(\vec{P})=\textstyle\sqrt{{\!\zeta_{k}}/{V_{k,l}^{(t)}}}\!\log_2(\!I_k(\vec{P})+p_k {g_{k,k}})\!\!-\! \sqrt{{\!\zeta_{k}}/{V_{k,l}^{(t)}}}\!\log_2(I_k(\vec{P})).
 \end{align}}
\begin{algorithm}[!t]
\caption{Proposed Resource Allocation Algorithm for Problem $\mathcal{P}$}
\label{algo:proposed}
\DontPrintSemicolon
\KwIn{$K$, $L$, $P_{\max}$, convergence threshold $\xi > 0$, utility matrix $[\Psi]_{k,l}$, Set of unmatched UEs $\Pi$}
\BlankLine
\textbf{Initialize:} ${P}^{(0)}={P_{\max}}/{K}$, $\boldsymbol{\Gamma} = \emptyset$, 
$\Pi \leftarrow \{1,\dots,K\}$

\BlankLine
\underline{\textbf{Phase 1: Power Optimization ($\mathcal{P}1$):}}\\
 \While{\big($\frac{\|\mathbf{P}^{(t)} - \mathbf{P}^{(t-1)}\| }{ \|\mathbf{P}^{(t-1)}\|} > \xi$\big)} {
 \For {$k \in K$}{
 Compute $I_k(\mathbf{P}^{(t)})=\sum_{ i\neq k}^K p_i^{(t)}g_{k,i}+\sigma_k^2$\\
 Compute $\Lambda_{k,l}^{(t)}$ using~\ref{SCA_1}\\
 Compute $V_{k,l}^{(t)}=1-{1}/{(1+\Lambda_{k,l}^{(t)})^2}$
 }
  \For {$k \in K$ and $i \in K$}{
 Compute SCA coefficient as $ \varrho_{k,i}$ using~\eqref{SCA_4}
 }
Solve~\eqref{eq_opt_prob1.2.1} to obtain $\mathbf{P}^{(t+1)}$\\
 $t\leftarrow t+1$
 }
\BlankLine
\underline{\textbf{Phase 2: UE-RIS Association ($\mathcal{P}2$):}}\\
\For{$k \in K$}{
    Compute priority list $[\vec{\Gamma}]_{\user_k,\IRS_l}$ and $[\vec{\Gamma}]_{\IRS_l,\user_k}$ using~\eqref{reliability}\;
}
\While{\big($\Pi \neq \emptyset$ \textbf{and} there exists at least one UE that has not applied to all RISs\big)}{
    \For{($k \in \Pi$)}{
        $\user_k$ proposes to its highest-ranked $\IRS_l$ in $[\vec{\Gamma}]_{\user_k,\IRS_l}$ that has not yet rejected it\;
        
        \uIf{\text{RIS is free}}{
            $\Pi \leftarrow \Pi \setminus \{k\}$\;
               $\Omega_{k,1} = 1$\;
        }
        \Else{
            $\IRS_l$ check its priority list\;
            
            \uIf{$[\vec{\Gamma}]_{\IRS_l,\user_k} > [\vec{\Gamma}]_{\IRS_l,\user_{k'}}$}{
                $\Pi \leftarrow (\Pi \setminus \{k\}) \cup \{k'\}$\;
                   $\Omega_{k,l} = 1$\;
            }
            \Else{
                $\IRS_l$ rejects $\user_k$\;
            }
        }
    }
}
\BlankLine
\KwOut{$\mathbf{P}^\star$ and $\boldsymbol{\Omega}^{\star}$\;}
\end{algorithm}
 The first term and second term in \eqref{SCA_2} both are concave functions, therefore, $\eta_{k,l}(\vec{P})$ is difference of two concave functions, which keeps~\eqref{SCA_2} nonconvex. Meanwhile, we utilize SCA, in which we keep the first term and linearize the second term at iteration $t$ using the first-order Taylor expansion, which can be written as
 \begin{align}\label{SCA_3}
&\textstyle\sqrt{{\zeta_{k}}/{V_{k,l}^{(t)}}}\log_2(I_k(\vec{P})\textstyle\leq\sqrt{{\zeta_{k}}/{V_{k,l}^{(t)}}}\log_2(I_k(\vec{P^{(t)}}))\nonumber\\&\quad\quad+\textstyle\sum_{i=1}^K \sqrt{{\zeta_{k}}/{V_{k,l}^{(t)}}}\varrho_{k,i}^{(t)}(p_i-p_i^{(t)}),
 \end{align}
where
{\small
\begin{align}\label{SCA_4}
 \textstyle  \varrho_{k,i}^{(t)} =
\begin{cases}
{g_{k,i}}/{\ln{2}\,I_k(\vec{P^{(t)}})}, &  i\neq k, \\
0, &  i= k.
\end{cases}
\end{align}}
Since $\log_2\big(I_k(\vec{P})\big)$ is concave and the first-order Taylor expansion in \eqref{SCA_3} provides an upper bound which guarantees monotonic improvement in SCA. Constant terms that are independent of $\vec{P}$ are excluded since they do not affect the optimization. Utilizing \eqref{SCA_2} and \eqref{SCA_3}, we can rewrite the optimization problem $\mathcal{P}1$ in \eqref{eq_opt_prob1.2}, can be written as
\begin{subequations}\label{eq_opt_prob1.2.1}
\begin{alignat}{2}
& \mathcal{P}1.1\!:\!\!\!\underset{\vec{\mathbf{P}}}{\text{ max}}
&&\sum_{k=1}^K\!\!\sqrt{\frac{\zeta_{k}}{V_{k,l}^{(t)}}}\big[\! \log_2(\sum_{ i=1}^K \!p_ig_{k,i}\!\!+\!\!\sigma_k^2)\!\!-\!\!\sum_{i=1}^K\! \varrho_{k,i}^{(t)}p_i\big], \label{eq_optProb1.2.3}\\ 
&\text{s.t.} 
 &&\quad \mathrm{C}1, \, \mathrm{C}2.\label{eq_constraint2.2.1}
\end{alignat}
\end{subequations}
\begin{Theorem}\label{theorem2}
    For a given RIS association matrix $\vec{\Omega}$, the transformed subproblem $\mathcal{P} 1.1$ is convex.
\end{Theorem}
\begin{IEEEproof}
From \eqref{variable_1} and \eqref{SCA_2}, the objective is given as
  {\small
  \begin{align}\label{SCA_2.1}
\textstyle\eta_{k,l}(\vec{P})\!\!=\!\!\textstyle\sqrt{\frac{\zeta_{k}}{V_{k,l}^{(t)}}}\log_2(I_k(\vec{P})\!+\!p_k {g_{k,k}})\!-\!\! \sqrt{\frac{\zeta_{k}}{V_{k,l}^{(t)}}}\log_2(I_k(\vec{P})),
 \end{align}}
which is a difference of two concave functions. Using the first-order Taylor approximation in~\eqref{SCA_3}, the second term is upper bounded as
{\small  \begin{align}\label{SCA_3.1}
&\textstyle\sqrt{{\zeta_{k}}/{V_{k,l}^{(t)}}}\log_2(I_k(\vec{P})\textstyle\leq\sqrt{{\zeta_{k}}/{V_{k,l}^{(t)}}}\log_2(I_k(\vec{P^{(t)}}))\nonumber\\&\quad\quad+\textstyle\sum_{i=1}^K \sqrt{{\zeta_{k}}/{V_{k,l}^{(t)}}}\varrho_{k,i}^{(t)}(p_i-p_i^{(t)})
 \end{align}}
 Substituting this bound into the objective yields a concave function. Since constraints $\mathrm{C}1$ and $\mathrm{C}2$ are convex, problem $\mathcal{P}1.1$ is convex.
\end{IEEEproof}
Based on Theorem~\ref {theorem2}, subproblem $\mathcal{P}1.1$ is convex and can be efficiently solved using standard convex solvers such as CVX, as described in Phase~1 of Algorithm~\ref{algo:proposed}.
\subsection{Deferred Acceptance Matching-based UE-RIS Association}
Then, we use the power allocation matrix obtained in subproblems $\mathcal{P}1.1$, to formulate the UE-RIS association as
\begin{subequations}\label{eq_opt_prob1.3}
\begin{alignat}{2}
& \mathcal{P}2: \underset{\vec{\Omega}}{\text{ max}}
& 
&\quad\quad\textstyle\sum_{k=1}^K \Omega_{k,l} \Psi_{k,l} (\vec{\Omega\,|\,P^{\star} }), \label{eq_optProb1.3}\\
&\text{s.t.} 
 &&\quad\quad  \mathrm{C}1-\mathrm{C}5.\label{eq_constraint5.3.2}
\end{alignat}
\end{subequations}
We deploy multiple RISs to assist the downlink communication. The UE-RIS association problem $\mathcal{P}2$ is formulated as a one-to-one matching problem, where each UE can be associated with at most one RIS within a transmission slot. As illustrated in Phase~2 of Algorithm~\ref{algo:proposed}, the association procedure follows a UE-proposing deferred-acceptance mechanism. The preference profiles of both UEs and RISs are constructed based on reliability-oriented metrics. Specifically, each UE ranks the RISs according to the achievable reliability, which is determined by the decoding success probability under finite blocklength transmission. Initially, all UEs are unmatched and maintain preference lists of RISs sorted in descending order of their achievable reliability. In each iteration, every unmatched UE proposes to its most preferred RIS that has not yet rejected it. Upon receiving one or more proposals, each RIS evaluates the reliability performance of the proposing UEs and tentatively accepts the UE that yields the highest reliability, while rejecting the others. If a RIS is already tentatively matched, it compares the new proposal with its current match and retains the UE that achieves higher reliability. Rejected UEs remove the corresponding RIS from their preference lists and proceed to propose to their next preferred RIS in subsequent iterations. This iterative process, as depicted in Phase~2 of Algorithm~\ref{algo:proposed}, continues until no further proposals can be made, at which point a stable UE-RIS matching is obtained that maximizes the overall network reliability.
\begin{Theorem}
    The proposed matching algorithm converges in a polynomial number of iterations. Specifically, for a network with K UEs and L RISs, it terminates after at most $K\times L$ matching operations. Moreover, the resulting association matrix $\vec{\Omega}^\star$ corresponds to a stable solution.
\end{Theorem}
\begin{IEEEproof}
The proposed algorithm converges in polynomial time. In each iteration, every unmatched UE proposes to an unvisited RIS based on its preference list, and each RIS accepts or rejects proposals accordingly. Since each UE can propose to at most L RISs, the total number of proposals is bounded by $K\times L$. Moreover, the resulting association matrix $\Omega$ converges to a stable solution. Specifically, if a UE prefers another RIS over its final match, it must have previously proposed to that RIS and been rejected in favor of a more preferred UE. Therefore, no blocking pair exists, ensuring the stability of the matching.
\end{IEEEproof}
\section{Simulation Results and Discussion} \label{pref}
We consider a network consisting of multiple UEs served by a single AP equipped with $64$ antennas deployed over the coverage area.  The communication is assisted by multiple RISs, where each RIS is equipped with $100\times100$ reflectors with half-wavelength spacing ($M_y=M_z= \lambda/2$). The carrier frequency and the channel bandwidth (BW) are set to be $f_c =15$ GHz and $200$ MHz, respectively. The AP transmit power budget is $23$ dBm, while the noise power density ($N_0$) and noise figure (NF) are $-174$ dBm/Hz, and $10$ dB, respectively. Accordingly, the AWGN noise power at each UE is given by $\sigma_k^2=N_0+10\log_{10}(BW)+NF$~dBm. All algorithms are implemented in MATLAB, and unless otherwise specified, each data point is averaged over $10^6$ independent channel realizations to ensure statistical accuracy. The performance of the proposed scheme is evaluated against three benchmark approaches, where the transmit power allocation is obtained using the same SCA-based optimization framework in all cases for fair comparison. In the first benchmark, referred to as SCA-ES, it iterates all possible combinations to maximize the network reliability. In the second benchmark, referred to as SCA-GS, each UE greedily selects the RIS that maximizes the reliability. In the event that multiple UEs select the same RIS, one UE is randomly retained while the others are rejected. In the third benchmark, denoted as SCA-RS, the UE-RIS association is generated randomly.
\begin{figure}[!htp]
     \centering
\includegraphics[width=0.9\linewidth]{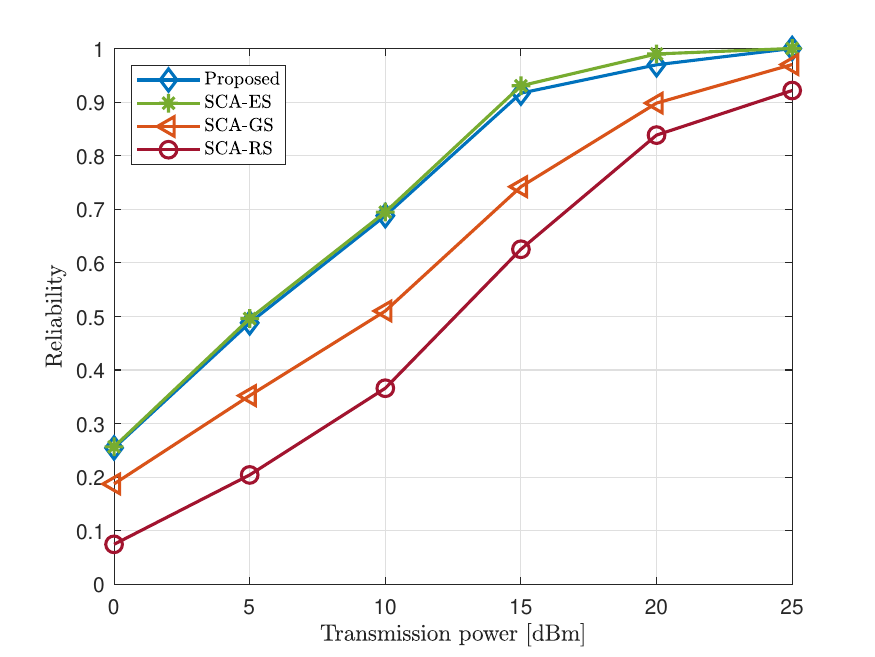}
\vspace{-0.4cm}\caption{Reliability as a function of the AP power budget.}\vspace{-0.2cm}
     \label{relibility_power} 
\end{figure}
The URLLC UE reliability as a function of the AP transmit power is illustrated in Fig.~\ref{relibility_power}. As expected, the reliability increases with transmit power for all considered schemes due to the enhanced received signal strength and improved SINR conditions. The proposed matching-based scheme achieves performance close to that of the SCA-ES benchmark and consistently outperforms the SCA-GS and SCA-RS baseline methods across the entire power range. This improvement is primarily attributed to the more efficient UE–RIS association enabled by the matching framework, which effectively exploits favorable channel conditions and available RIS resources. As a result, the proposed approach enables improved interference management and resource utilization, leading to enhanced overall reliability.
\begin{figure}[!htp]
     \centering
\includegraphics[width=0.9\linewidth]{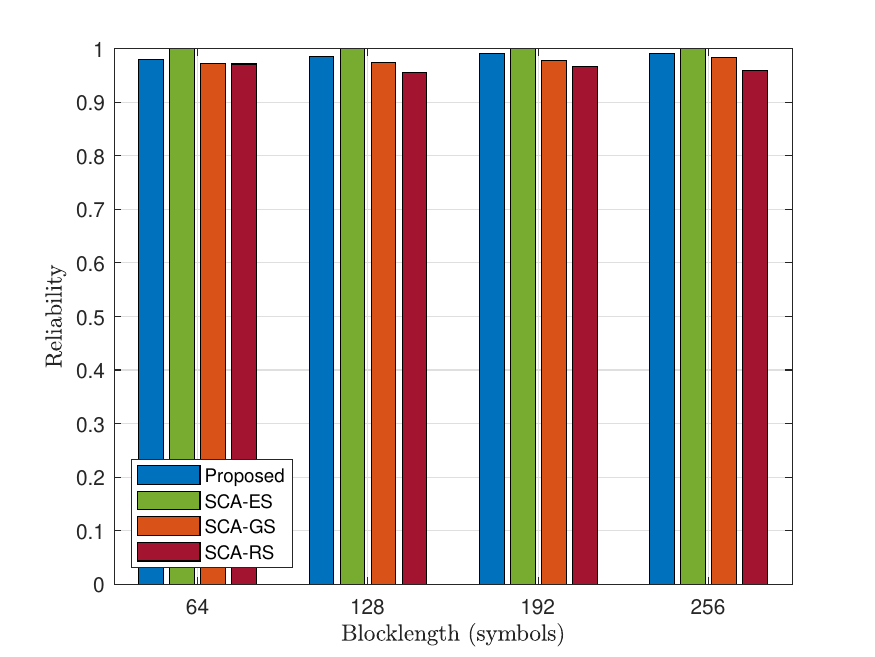}
   \vspace{-0.4cm} \caption{Reliability as a function of the finite blocklength.}\vspace{-0.2cm}
\label{relibility_blocklength}   
\end{figure}
The URLLC UE reliability as a function of the finite blocklength is illustrated in Fig.~\ref{relibility_blocklength}. As expected, the reliability improves with increasing blocklength for all considered schemes due to the enhanced coding gain and reduced decoding error probability. The proposed matching-based scheme achieves performance close to the SCA-ES benchmark and consistently outperforms the SCA-GS and SCA-RS baseline methods across all blocklength values, due to efficient UE-RIS association. Consequently, the proposed approach ensures improved resource utilization and interference management, leading to higher overall reliability.
\begin{figure}[!t]
     \centering
\includegraphics[width=0.9\linewidth]{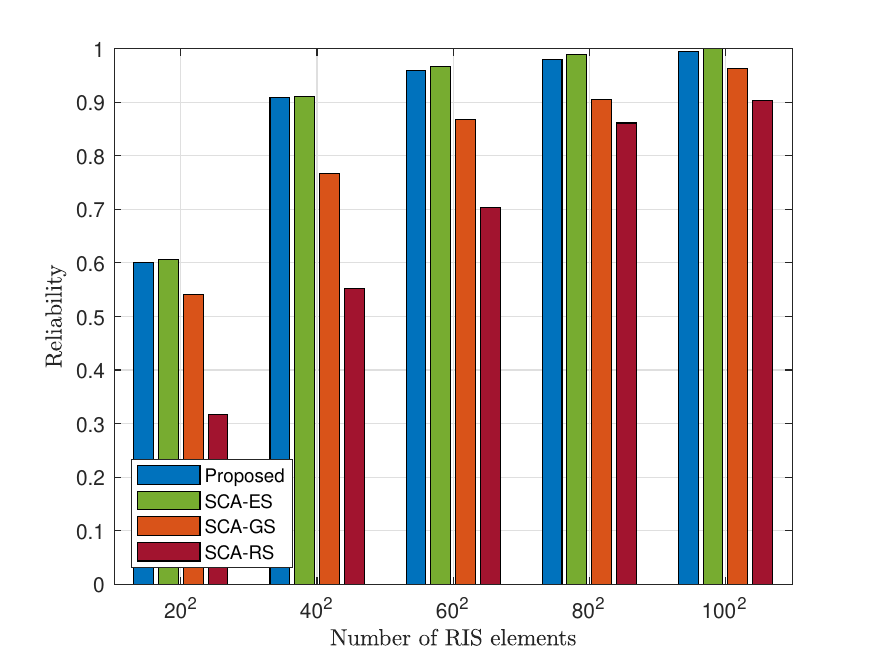}
   \vspace{-0.4cm} \caption{Reliability as a function of the number of RIS elements.}\vspace{-0.2cm}
\label{reliability_elements}  
\end{figure}
 The reliability as a function of the number of RIS elements is depicted in Fig.~\ref{reliability_elements}. As the number of reflecting elements increases, the reliability improves for all schemes due to the enhanced passive beamforming gain and more effective signal reflection. The proposed matching-based approach consistently outperforms the SCA-GS and SCA-RS baselines across all configurations and achieves performance close to that of the SCA-ES scheme. This improvement stems from the more efficient UE-RIS association provided by the matching framework, which aligns users with RISs offering more favorable channel conditions and resource availability.
\section{Conclusions}\label{conc}
This paper investigated reliability-oriented resource allocation in RIS-assisted FR3 networks. A joint transmit power allocation and UE-RIS association problem is formulated to maximize network reliability under finite blocklength constraints, capturing the stringent requirements of URLLC services. Due to the nonconvex and combinatorial nature of the problem, a multiphase solution framework is developed, integrating an SCA-based power optimization scheme with a matching-theory-based UE-RIS association algorithm. Simulation results demonstrated that the proposed approach achieves improved reliability performance compared with conventional greedy and random association schemes, while efficiently utilizing available resources. These results highlight the importance of jointly optimizing power allocation and RIS association in FR3 networks. Future work will extend this framework to incorporate active user detection and dynamic resource allocation for grant-free massive access scenarios. Specifically, jointly optimizing user activity detection, power control, and UE-RIS association under time-varying network conditions is a promising direction to improve scalability and robustness in RIS-assisted URLLC systems. 
\bibliographystyle{IEEEtran}
\bibliography{References}
\end{document}